\documentclass[conference,a4paper]{IEEEtran}%
\usepackage[utf8]{inputenc}
\usepackage[T1]{fontenc}
\usepackage{url}
\usepackage{ifthen}
\usepackage{cite}
\usepackage[cmex10]{amsmath}
\usepackage{amsfonts}
\usepackage{amssymb}
\usepackage{hyperref}
\usepackage{graphicx}%
\providecommand{\U}[1]{\protect\rule{.1in}{.1in}}
\newtheorem{theorem}{Theorem}

\newtheorem{lemma}[theorem]{Lemma}

\newenvironment{proof}[1][Proof]{\noindent\textbf{#1.} }{\ \rule{0.5em}{0.5em}}
\begin{document}

\title{Recoverability for Holevo's just-as-good fidelity}
\author{\IEEEauthorblockN{Mark M.~Wilde}
\IEEEauthorblockA{Hearne Institute for Theoretical Physics,
Department of Physics and Astronomy,
Center for Computation and Technology,\\
Louisiana State University,
Baton Rouge, Louisiana 70803, USA,
Email: mwilde@lsu.edu} }
\maketitle

\begin{abstract}
Holevo's just-as-good fidelity is a similarity measure for quantum states that
has found several applications. One of its critical
properties is that it obeys a data processing inequality:\ the measure does
not decrease under the action of a quantum channel on the underlying states.
In this paper, I prove a refinement of this data processing inequality that
includes an additional term related to recoverability. That is, if the increase in the
measure is small after the action of a partial trace, then one of the states
can be nearly recovered by the Petz recovery channel, while the other state is
perfectly recovered by the same channel. The refinement is given in terms of
the trace distance of one of the states to its recovered version and also
depends on the minimum eigenvalue of the other state. As such, the refinement
is universal, in the sense that the recovery channel depends only on one of
the states, and it is explicit, given by the Petz recovery channel. The appendix contains a generalization of the aforementioned result to arbitrary quantum channels.

\end{abstract}

\section{Introduction}

In Holevo's seminal 1972 work on the quasiequivalence of locally normal states
\cite{Kholevo1972}, he established the following inequalities for quantum
states $\rho$ and$~\sigma$:%
\begin{equation}
1-\sqrt{F_{H}(\rho,\sigma)}\leq\frac{1}{2}\left\Vert \rho-\sigma\right\Vert
_{1}\leq\sqrt{1-F_{H}(\rho,\sigma)}, \label{eq:holevo-ineq}%
\end{equation}
where $\left\Vert \rho-\sigma\right\Vert _{1}$ denotes the well known trace
distance and the function $F_{H}$ is Holevo's \textquotedblleft just-as-good
fidelity,\textquotedblright\ defined as%
\begin{equation}
F_{H}(\rho,\sigma)\equiv\left[  \operatorname{Tr}\{\sqrt{\rho}\sqrt{\sigma
}\}\right]  ^{2}.
\end{equation}
After writing it down, he then remarked that \textquotedblleft it is evident
that $F_{H}$ is just as good a measure of proximity of the states $\rho$ and
$\sigma$ as $\left\Vert \rho-\sigma\right\Vert _{1}$.\textquotedblright\ And
so it is that the measure $F_{H}$ is known as Holevo's just-as-good fidelity.

Some years after this, Uhlmann defined the quantum fidelity as 
$
F(\rho,\sigma)\equiv\left\Vert \sqrt{\rho}\sqrt{\sigma}\right\Vert _{1}^{2}$ \cite{U76}.
It is evident that the following relation holds%
\begin{equation}
F_{H}(\rho,\sigma)\leq F(\rho,\sigma),
\label{eq:F_H-to-F}
\end{equation}
due to the variational characterization of the trace norm of a square operator
$X$ as%
\begin{equation}
\left\Vert X\right\Vert _{1}=\max_{U}\left\vert \operatorname{Tr}%
\{XU\}\right\vert , \label{eq:tr-norm-var}%
\end{equation}
where the optimization is with respect to a unitary operator $U$. Many years
after this, at the dawn of quantum computing, with more growing interest in
quantum information theory, Fuchs and van de Graaf presented the following
widely employed inequalities \cite{FG98}:%
\begin{equation}
1-\sqrt{F(\rho,\sigma)}\leq\frac{1}{2}\left\Vert \rho-\sigma\right\Vert
_{1}\leq\sqrt{1-F(\rho,\sigma)}, \label{eq:FvdG}%
\end{equation}
which bear a striking similarity to \eqref{eq:holevo-ineq}. Indeed the lower
bound on $\frac{1}{2}\left\Vert \rho-\sigma\right\Vert _{1}$ in
\eqref{eq:FvdG} is an immediate consequence of \eqref{eq:F_H-to-F} and the lower bound in
\eqref{eq:holevo-ineq}. The upper bound on $\frac
{1}{2}\left\Vert \rho-\sigma\right\Vert _{1}$ in \eqref{eq:FvdG} can be proven
by first showing that it is achieved for pure states, employing Uhlmann's
\textquotedblleft transition probability\textquotedblright\ characterization
of $F(\rho,\sigma)$ \cite{U76}, and then invoking monotonicity of trace
distance with respect to partial trace. The latter inequalities in
\eqref{eq:FvdG} have been more widely employed in quantum information theory
than those in \eqref{eq:holevo-ineq} due to Uhlmann's \textquotedblleft
transition probability\textquotedblright\ characterization of $F(\rho,\sigma)$
and its many implications.

Nevertheless, Holevo's just-as-good fidelity is clearly a useful measure of
similarity for quantum states in light of \eqref{eq:holevo-ineq}, and it has
found several applications in quantum information theory. For example, it
serves as an upper bound on the probability of error in discriminating $\rho$
from $\sigma$ in a hypothesis testing experiment \cite{ACMBMAV07,CMMAB08},
which in some sense is just a rewriting of the lower bound in
\eqref{eq:holevo-ineq} (see also \cite[Lemma~3.2]{H06book} in this context).
In turn, this way of thinking has led to particular decoders for quantum polar codes
\cite{WG11,GW12}.

The function $F_H$ has also been rediscovered a number of times. For example, it is a
particular case of Petz's quasi-entropies \cite{P85,P86}. It was studied under
the name \textquotedblleft quantum affinity\textquotedblright\ in
\cite{LZ04}\ and shown to be equal to the fidelity of the canonical
purifications of quantum states in \cite{Winter01a}.

One of the most important properties of $F_{H}$ is that it obeys the following
data processing inequality:%
\begin{equation}
F_{H}(\mathcal{N}(\rho),\mathcal{N}(\sigma))\geq F_{H}(\rho,\sigma),
\label{eq:fid-data-proc}%
\end{equation}
where $\mathcal{N}$ is a quantum channel (a completely positive and trace
preserving map). This inequality is a
consequence of data processing for Petz's more general quasi-entropies
\cite{P85,P86}. This property is one reason that $F_{H}$ has an interpretation
as a similarity measure: the states $\rho$ and $\sigma$ generally  become
more similar under the action of a quantum channel.

The main contribution of this paper is the following refinement of the data
processing inequality in \eqref{eq:fid-data-proc}, in the case that $\rho$ is a bipartite density operator, $\sigma$ is a positive definite bipartite operator, and the channel is
a partial trace over the $B$ system:%
\begin{multline}
\sqrt{F_{H}}(\rho_{A},\sigma
_{A})\geq
\sqrt{F_{H}}(\rho_{AB},\sigma_{AB})
\label{eq:main-result}\\
+\frac{\pi^{2}}{432} \frac{\lambda_{\min}(\sigma_{AB})}{\operatorname{Tr}\{\sigma_{A}%
\}}\left\Vert \mathcal{R}_{A\rightarrow AB}^{\sigma}(\rho_{A})-\rho
_{AB}\right\Vert _{1}^{3},
\end{multline}
where $\lambda_{\min}(\sigma_{AB})$ is the minimum eigenvalue of
$\sigma_{AB}$ and%
\begin{equation}
\mathcal{R}_{A\rightarrow AB}^{\sigma}(\cdot)\equiv\sigma_{AB}^{1/2}\left[
\sigma_{A}^{-1/2}(\cdot)_{A}\sigma_{A}^{-1/2}\otimes I_{B}\right]  \sigma
_{AB}^{1/2}%
\end{equation}
is a quantum channel known as the Petz recovery channel
\cite{Petz1986,Petz1988}. The interpretation of this inequality is the same as
that given in previous work on this topic of refining data processing
inequalities (see, e.g., \cite{FR14,W15,Junge15}). If the difference
$\sqrt{F_{H}}(\rho_{A},\sigma_{A})-\sqrt{F_{H}}(\rho_{AB},\sigma_{AB})$ is
small, then one can approximately recover the state $\rho_{AB}$ from its
marginal $\rho_{A}$.
The appendix generalizes the result in
\eqref{eq:main-result} to arbitrary quantum channels.

The technique that I use for proving the above data processing refinement
closely follows the elegant approach recently put forward by Carlen and Vershynina in
\cite{CV17}. This technique appears to be different from every other approach, given
in recent years since \cite{FR14}, that has established refinements of data
processing inequalities. It builds on Petz's approach from \cite{P85,P86}\ for
proving data processing for the quantum relative entropy, as well as ideas in
\cite{Petz03}. Here, I use this same technique and establish a general lemma
regarding remainder terms for data processing with Petz's quasi-entropies, and
then I\ specialize it to obtain the inequality in \eqref{eq:main-result}.

An interesting aspect of  \eqref{eq:main-result} is that the recovery
channel is explicit, given in the Petz form, and universal, having no
dependence on the state $\rho_{AB}$ while depending only on $\sigma_{AB}$.

In the rest of the paper, I begin by giving background and establish some
notation. After that, I prove a general lemma that refines data processing for
Petz's quasi-entropies. Then I specialize it to arrive at the inequality in
\eqref{eq:main-result}.

\section{Background and Notation}

I begin by reviewing some background and establish notation. Basic concepts of
quantum information theory can be found in \cite{H06book,H12,W17}. Let $f$ be
an operator convex function defined on $[0,\infty)$. Examples include
$f(x)=x\ln x$, $f(x)=-x^{\alpha}$ for $\alpha\in(0,1)$, $f(x)=x^{\alpha}$ for
$\alpha\in(1,2]$. According to \cite[Section~8]{HMPB11}, such a function has
the following integral representation:%
\begin{multline}
f(x)=f(0)+ax+bx^{2}\\
+\int_{0}^{\infty}d\mu(t)\ \left(  \frac{x}{1+t}-1+\frac{t}{x+t}\right)  ,
\label{eq:integral-rep}%
\end{multline}
where $a\in\mathbb{R}$, $b\geq0$, and $\mu$ is a non-negative measure on
$(0,\infty)$ satisfying
$
\int_{0}^{\infty}
d\mu(t) / \left(  1+t\right)  ^{2}<\infty$.

Define the maximally entangled vector as
\begin{equation}
\left\vert \Gamma\right\rangle _{S\tilde{S}}\equiv\sum_{i=0}^{\left\vert
S\right\vert -1}\left\vert i\right\rangle _{S}\left\vert i\right\rangle
_{\tilde{S}},
\end{equation}
for orthonormal bases $\{\left\vert i\right\rangle _{S}\}_{i}$ and
$\{\left\vert i\right\rangle _{\tilde{S}}\}_{i}$, and for a positive
semi-definite operator $\sigma$, define its canonical purification by%
\begin{equation}
\left\vert \varphi^{\sigma}\right\rangle _{S\tilde{S}}\equiv\left(  \sigma
_{S}^{1/2}\otimes I_{\tilde{S}}\right)  \left\vert \Gamma\right\rangle
_{S\tilde{S}}.
\end{equation}
Then, following Petz \cite{P85,P86,P10,P10a}, as well as what was
discussed later in \cite{TCR09,S10}, we define the $f$-quasi-relative entropy
$Q_{f}(\rho\Vert\sigma)$ of a density operator $\rho$ and a positive
definite operator $\sigma$ as%
\begin{equation}
Q_{f}(\rho\Vert\sigma)\equiv\left\langle \varphi^{\sigma}\right\vert
_{S\tilde{S}}f\left(  \sigma_{S}^{-1}\otimes\rho_{\tilde{S}}^{T}\right)
\left\vert \varphi^{\sigma}\right\rangle _{S\tilde{S}}.
\end{equation}
For example, when $f(x)=x\ln x$, then $Q_{f}$ reduces to the quantum relative
entropy\ from \cite{U62}.

Now consider the bipartite case and define%
\begin{equation}
\left\vert \Gamma\right\rangle _{A\hat{A}B\hat{B}}\equiv\left\vert
\Gamma\right\rangle _{A\hat{A}}\otimes\left\vert \Gamma\right\rangle
_{B\hat{B}}.
\end{equation}
We can also write this as $\left\vert \Gamma\right\rangle _{AB\hat{A}\hat{B}}$
with it being understood that there is a permutation of systems. Then, by the
above, we have for a density operator $\rho_{AB}$ and a positive definite
operator $\sigma_{AB}$ that%
\begin{multline}
Q_{f}(\rho_{AB}\Vert\sigma_{AB})\\
=\left\langle \varphi^{\sigma_{AB}}\right\vert _{AB\hat{A}\hat{B}}f\!\left(
\sigma_{AB}^{-1}\otimes\rho_{\hat{A}\hat{B}}^{T}\right)  \left\vert
\varphi^{\sigma_{AB}}\right\rangle _{AB\hat{A}\hat{B}}.
\end{multline}
Now define the linear operator $V$ by%
\begin{equation}
V_{A\hat{A}\rightarrow AB\hat{A}\hat{B}}\equiv\sigma_{AB}^{1/2}\left(
\sigma_{A}^{-1/2}\otimes I_{\hat{A}}\right)  \left\vert \Gamma\right\rangle
_{B\hat{B}}.
\end{equation}
This linear operator is an isometric extension of the Petz recovery channel, as
discussed recently in \cite{W17f}. One can readily verify that $V$\ is an isometry and that%
\begin{align}
V_{A\hat{A}\rightarrow AB\hat{A}\hat{B}}\left\vert \varphi^{\sigma_{A}%
}\right\rangle _{A\hat{A}} & =\left\vert \varphi^{\sigma_{AB}}\right\rangle
_{AB\hat{A}\hat{B}},\\
V^{\dag}\left(  \sigma_{AB}^{-1}\otimes\rho_{\hat{A}\hat{B}}^{T}\right)
V & =\sigma_{A}^{-1}\otimes\rho_{\hat{A}}^{T}.\label{eq:V-sandw}%
\end{align}
For simple proofs of these properties, see, e.g., \cite{TCR09} or \cite{W17f}.
With all these notions in place, we can recall Petz's approach
\cite{P85,P86,P10,P10a}\ for establishing monotonicity of the $f$%
-quasi-relative entropy under partial trace:%
\begin{align}
&  Q_{f}(\rho_{AB}\Vert\sigma_{AB})
  =\left\langle \varphi^{\sigma_{AB}}\right\vert _{AB\hat{A}\hat{B}}f\!\left(
\sigma_{AB}^{-1}\otimes\rho_{\hat{A}\hat{B}}^{T}\right)  \left\vert
\varphi^{\sigma_{AB}}\right\rangle _{AB\hat{A}\hat{B}}\nonumber \\
&  =\left\langle \varphi^{\sigma_{A}}\right\vert _{A\hat{A}}V^{\dag}f\!\left(
\sigma_{AB}^{-1}\otimes\rho_{\hat{A}\hat{B}}^{T}\right)  V\left\vert
\varphi^{\sigma_{A}}\right\rangle _{A\hat{A}}\nonumber\\
&  \geq\left\langle \varphi^{\sigma_{A}}\right\vert _{A\hat{A}}f\!\left(
V^{\dag}\left[  \sigma_{AB}^{-1}\otimes\rho_{\hat{A}\hat{B}}^{T}\right]
V\right)  \left\vert \varphi^{\sigma_{A}}\right\rangle _{A\hat{A}}\nonumber\\
&  =\left\langle \varphi^{\sigma_{A}}\right\vert _{A\hat{A}}f\!\left(
\sigma_{A}^{-1}\otimes\rho_{\hat{A}}^{T}\right)  \left\vert \varphi
^{\sigma_{A}}\right\rangle _{A\hat{A}}
  =Q_{f}(\rho_{A}\Vert\sigma_{A})
\end{align}
where we made use of everything above and the operator Jensen inequality
\cite{HP03}.

\section{General statement for quasi-entropies}

I now modify the approach from \cite{CV17}\ for lower bounds for relative
entropy differences to use an arbitrary operator convex function $f$
instead. So we are considering the following $f$-quasi-relative entropy
difference:%
\begin{equation}
Q_{f}(\rho_{AB}\Vert\sigma_{AB})-Q_{f}(\rho_{A}\Vert\sigma_{A}).
\end{equation}
Recall the integral representation of $f$ from \eqref{eq:integral-rep}. Let%
\begin{align}
\Delta_{AB\hat{A}\hat{B}}  &  \equiv\sigma_{AB}^{-1}\otimes\rho_{\hat{A}%
\hat{B}}^{T}, \qquad
\Delta_{A\hat{A}}    \equiv\sigma_{A}^{-1}\otimes\rho_{\hat{A}}^{T},\\
V_{A\hat{A}\rightarrow AB\hat{A}\hat{B}}  &  \equiv\sigma_{AB}^{1/2}\left(
\sigma_{A}^{-1/2}\otimes I_{\hat{A}}\right)  \left\vert \Gamma\right\rangle
_{B\hat{B}}%
\end{align}
and recall from \eqref{eq:V-sandw}\ that
$
V^{\dag}\Delta_{AB\hat{A}\hat{B}}V=\Delta_{A\hat{A}}$. This implies
\begin{align}
\left \Vert \Delta_{A\hat{A}} \right \Vert_{\infty}
& =
\left\Vert V^\dag \Delta_{AB\hat{A}\hat{B}} V \right \Vert_{\infty} = 
\left\Vert V V^\dag \Delta_{AB\hat{A}\hat{B}} V V^\dag \right \Vert_{\infty}
\nonumber \\
& \leq \left\Vert  \Delta_{AB\hat{A}\hat{B}} \right \Vert_{\infty}
\end{align}
with the last equality following from isometric invariance of the operator norm
and the inequality from submultiplicativity of the operator norm and the fact that $VV^\dag$ is a projection.

\begin{lemma}
\label{lem:quasi}
Let $\mu$ be a measure. For an operator $X$, define%
\begin{equation}
\nu(X)=\int_{0}^{\infty}\text{d}\mu(t)~t\left(  \frac{1}{t}-\frac{1}%
{t+X}\right)  , \label{eq:mu-int}%
\end{equation}
and for $T>0$, define $\mu([0,T])\equiv\int_{0}^{T}$d$\mu(t)$. For $c>0$,
define $g(c,T)\equiv\int_{T}^{\infty}$d$\mu(t)\ \frac{1}{1+t/c}$. Let
$\rho_{AB}$ be a density operator and $\sigma_{AB}$ a positive definite
operator. Then for all $T>0$, the following inequality holds
\begin{multline}
\left\Vert \left[  \sigma_{AB}^{1/2}\sigma_{A}^{-1/2}\nu(\Delta_{A\hat{A}%
})\sigma_{A}^{1/2}-\nu(\Delta_{AB\hat{A}\hat{B}})\sigma_{AB}^{1/2}\right]
\left\vert \Gamma\right\rangle _{A\hat{A}B\hat{B}}\right\Vert _{2}%
\label{eq:intermediate}\\
\leq\mu([0,T])^{1/2}\left[  Q_{f}(\rho_{AB}\Vert\sigma_{AB})-Q_{f}(\rho
_{A}\Vert\sigma_{A})\right]  ^{1/2}\\
+2g(\left\Vert \Delta_{AB\hat{A}\hat{B}}\right\Vert_{\infty} ,T)\operatorname{Tr}%
\{\sigma_{A}\}.
\end{multline}

\end{lemma}

\begin{proof}
The proof follows \cite{CV17} quite closely at times but also features some departures.
Since $V$ is an isometry satisfying $V^{\dag}V=I_{A\hat{A}}$, it follows that
$VV^{\dag}$ is a projection, so that $VV^{\dag}\leq I_{AB\hat{A}\hat{B}}$.
Using the integral representation in \eqref{eq:integral-rep}, we arrive at the chain of inequalities in \eqref{eq:chain-1}\begin{figure*}[ptb]%
\begin{align}
  Q_{f}(\rho_{AB}\Vert\sigma_{AB})
&  =\left\langle \varphi^{\sigma_{A}}\right\vert _{A\hat{A}}V^{\dag}\left[
f\!\left(  \Delta_{AB\hat{A}\hat{B}}\right)  \right]  V\left\vert \varphi
^{\sigma_{A}}\right\rangle _{A\hat{A}}
\nonumber \\
&  =\left\langle \varphi^{\sigma_{A}}\right\vert _{A\hat{A}}V^{\dag}\left[
f(0)+a\Delta_{AB\hat{A}\hat{B}}+b\Delta_{AB\hat{A}\hat{B}}^{2}+\int
_{0}^{\infty}d\mu(t)\ \left(  \frac{\Delta_{AB\hat{A}\hat{B}}}{1+t}-1+\frac
{t}{\Delta_{AB\hat{A}\hat{B}}+t}\right)  \right]  V\left\vert \varphi
^{\sigma_{A}}\right\rangle _{A\hat{A}}
\nonumber \\
&  =f(0)+\left\langle \varphi^{\sigma_{A}}\right\vert _{A\hat{A}}\left[
aV^{\dag}\Delta_{AB\hat{A}\hat{B}}V+bV^{\dag}\Delta_{AB\hat{A}\hat{B}}%
^{2}V\right]  \left\vert \varphi^{\sigma_{A}}\right\rangle _{A\hat{A}%
}\nonumber\\
&  \qquad+\left\langle \varphi^{\sigma_{A}}\right\vert _{A\hat{A}}\left[
\int_{0}^{\infty}d\mu(t)\ \left(  \frac{V^{\dag}\Delta_{AB\hat{A}\hat{B}}%
V}{1+t}-1+V^{\dag}\frac{t}{\Delta_{AB\hat{A}\hat{B}}+t}V\right)  \right]
\left\vert \varphi^{\sigma_{A}}\right\rangle _{A\hat{A}}\nonumber \\
&  \geq f(0)+\left\langle \varphi^{\sigma_{A}}\right\vert _{A\hat{A}}\left[
aV^{\dag}\Delta_{AB\hat{A}\hat{B}}V+bV^{\dag}\Delta_{AB\hat{A}\hat{B}}%
VV^{\dag}\Delta_{AB\hat{A}\hat{B}}V\right]  \left\vert \varphi^{\sigma_{A}%
}\right\rangle _{A\hat{A}}\nonumber\\
&  \qquad+\left\langle \varphi^{\sigma_{A}}\right\vert _{A\hat{A}}\left[
\int_{0}^{\infty}d\mu(t)\ \left(  \frac{V^{\dag}\Delta_{AB\hat{A}\hat{B}}%
V}{1+t}-1+V^{\dag}\frac{t}{\Delta_{AB\hat{A}\hat{B}}+t}V\right)  \right]
\left\vert \varphi^{\sigma_{A}}\right\rangle _{A\hat{A}}\nonumber \\
&  =f(0)+\left\langle \varphi^{\sigma_{A}}\right\vert _{A\hat{A}}\left[
a\Delta_{A\hat{A}}+b\Delta_{A\hat{A}}^{2}+\int_{0}^{\infty}d\mu(t)\ \left(
\frac{\Delta_{A\hat{A}}}{1+t}-1+V^{\dag}\left(  \frac{t}{\Delta_{AB\hat{A}%
\hat{B}}+t}\right)  V\right)  \right]  \left\vert \varphi^{\sigma_{A}%
}\right\rangle _{A\hat{A}},
\label{eq:chain-1}
\end{align}
\hrulefill
\end{figure*}, where we made use of \eqref{eq:V-sandw} and the
fact that $VV^{\dag}$ is a projection so that $VV^{\dag}\leq I_{AB\hat{A}%
\hat{B}}$. Similarly, we find that%
\begin{align}
&  Q_{f}(\rho_{A}\Vert\sigma_{A})  =\left\langle \varphi^{\sigma_{A}}\right\vert _{A\hat{A}}\left[  f\!\left(
\Delta_{A\hat{A}}\right)  \right]  \left\vert \varphi^{\sigma_{A}%
}\right\rangle _{A\hat{A}} \nonumber\\
&  =f(0)+\left\langle \varphi^{\sigma_{A}}\right\vert _{A\hat{A}}\left[
a\Delta_{A\hat{A}}+b\Delta_{A\hat{A}}^{2}\right]  \left\vert \varphi
^{\sigma_{A}}\right\rangle _{A\hat{A}}\nonumber\\
&  +\int_{0}^{\infty}d\mu(t) \left\langle \varphi^{\sigma_{A}}\right\vert
_{A\hat{A}}\left(  \frac{\Delta_{A\hat{A}}}{1+t}-1+\frac{t}{\Delta_{A\hat{A}%
}+t}\right)  \left\vert \varphi^{\sigma_{A}}\right\rangle _{A\hat{A}}.
\end{align}
Thus, we find that%
\begin{multline}
Q_{f}(\rho_{AB}\Vert\sigma_{AB})-Q_{f}(\rho_{A}\Vert\sigma_{A})\geq\\
\int_{0}^{\infty}d\mu(t)\ t \left\langle \varphi^{\sigma_{A}}\right\vert
_{A\hat{A}}\Bigg[V^{\dag}\left(  \Delta_{AB\hat{A}\hat{B}}+t\right)  ^{-1}V\\
-\left(  \Delta_{A\hat{A}}+t\right)  ^{-1}\Bigg]\left\vert \varphi^{\sigma
_{A}}\right\rangle _{A\hat{A}}.
\end{multline}
Now consider that for $t > 0$%
\begin{align}
&  t \left\langle \varphi^{\sigma_{A}}\right\vert _{A\hat{A}}\left[  V^{\dag
}(  \Delta_{AB\hat{A}\hat{B}}+t)  ^{-1}V-\left(  \Delta_{A\hat{A}%
}+t\right)  ^{-1}\right]  \left\vert \varphi^{\sigma_{A}}\right\rangle
_{A\hat{A}}\nonumber\\
&  =t\, \langle\varphi^{w^{t}}|_{AB\hat{A}\hat{B}}\left(  \Delta_{AB\hat{A}%
\hat{B}}+t\right)  |\varphi^{w^{t}}\rangle_{AB\hat{A}\hat{B}}\nonumber\\
&  \geq t^{2} \left\Vert |\varphi^{w^{t}}\rangle_{AB\hat{A}\hat{B}%
}\right\Vert _{2}^{2},
\end{align}
where%
\begin{multline}
|\varphi^{w^{t}}\rangle_{AB\hat{A}\hat{B}}\equiv V\left(  \Delta_{A\hat{A}%
}+t\right)  ^{-1}\left\vert \varphi^{\sigma_{A}}\right\rangle _{A\hat{A}}\\
-\left(  \Delta_{AB\hat{A}\hat{B}}+t\right)  ^{-1}\left\vert \varphi
^{\sigma_{AB}}\right\rangle _{AB\hat{A}\hat{B}}.
\end{multline}
Consider that%
\begin{align}
|\varphi^{w^{t}}\rangle_{AB\hat{A}\hat{B}}  &  =\sigma_{AB}^{1/2}\sigma
_{A}^{-1/2}\left(  \Delta_{A\hat{A}}+t\right)  ^{-1}\sigma_{A}^{1/2}\left\vert
\Gamma\right\rangle _{A\hat{A}}\left\vert \Gamma\right\rangle _{B\hat{B}%
}\nonumber\\
&  \quad-\left(  \Delta_{AB\hat{A}\hat{B}}+t\right)  ^{-1}\sigma_{AB}%
^{1/2}\left\vert \Gamma\right\rangle _{A\hat{A}}\left\vert \Gamma\right\rangle
_{B\hat{B}}\\
&  =\bigg[\sigma_{AB}^{1/2}\sigma_{A}^{-1/2}\left(  \Delta_{A\hat{A}%
}+t\right)  ^{-1}\sigma_{A}^{1/2}\nonumber\\
&  \quad-\left(  \Delta_{AB\hat{A}\hat{B}}+t\right)  ^{-1}\sigma_{AB}%
^{1/2}\bigg]\left\vert \Gamma\right\rangle _{A\hat{A}}\left\vert
\Gamma\right\rangle _{B\hat{B}}.
\end{align}
So we set%
\begin{multline}
w_{AB\hat{A}\hat{B}}^{t}\equiv\sigma_{AB}^{1/2}\sigma_{A}^{-1/2}\left(
\Delta_{A\hat{A}}+t\right)  ^{-1}\sigma_{A}^{1/2}\\
-\left(  \Delta_{AB\hat{A}\hat{B}}+t\right)  ^{-1}\sigma_{AB}^{1/2},
\end{multline}
so that%
\begin{equation}
|\varphi^{w^{t}}\rangle_{AB\hat{A}\hat{B}}=w_{AB\hat{A}\hat{B}}^{t}\left\vert
\Gamma\right\rangle _{A\hat{A}B\hat{B}}.
\end{equation}
Now invoking the definition in \eqref{eq:mu-int}\ we find that%
\begin{align}
&  \sigma_{AB}^{1/2}\sigma_{A}^{-1/2}\nu(\Delta_{A\hat{A}})\sigma_{A}%
^{1/2}-\nu(  \Delta_{AB\hat{A}\hat{B}})  \sigma_{AB}^{1/2}%
\nonumber\\
&  =\int_{0}^{\infty}\text{d}\mu(t)~t\ \sigma_{AB}^{1/2}\sigma_{A}%
^{-1/2}\left(  \frac{1}{t}-\frac{1}{t+\Delta_{A\hat{A}}}\right)  \sigma
_{A}^{1/2}\nonumber\\
&  \qquad-\int_{0}^{\infty}\text{d}\mu(t)~t\left(  \frac{1}{t}-\frac
{1}{t+\Delta_{AB\hat{A}\hat{B}}}\right)  \sigma_{AB}^{1/2}\\
&  =\int_{0}^{\infty}\text{d}\mu(t)~t\ \Bigg[\sigma_{AB}^{1/2}\sigma
_{A}^{-1/2}\left(  \frac{1}{t}-\frac{1}{t+\Delta_{A\hat{A}}}\right)
\sigma_{A}^{1/2}\nonumber\\
&  \qquad-\left(  \frac{1}{t}-\frac{1}{t+\Delta_{AB\hat{A}\hat{B}}}\right)
\sigma_{AB}^{1/2}\Bigg]\\
&  =\int_{0}^{\infty}\text{d}\mu(t)~t\ \Bigg[-\sigma_{AB}^{1/2}\sigma
_{A}^{-1/2}\frac{1}{t+\Delta_{A\hat{A}}}\sigma_{A}^{1/2}\nonumber\\
&  \qquad+\frac{1}{t+\Delta_{AB\hat{A}\hat{B}}}\sigma_{AB}^{1/2}\Bigg]\\
&  =-\int_{0}^{\infty}\text{d}\mu(t)~tw_{AB\hat{A}\hat{B}}^{t}.
\end{align}
Thus, for any $T>0$, we have that%
\begin{align}
&  \left\Vert \left[\sigma_{AB}^{1/2}\sigma_{A}^{-1/2}\nu(\Delta_{A\hat{A}}%
)\sigma_{A}^{1/2}-\nu(  \Delta_{AB\hat{A}\hat{B}})  \sigma
_{AB}^{1/2}\right]\left\vert \Gamma\right\rangle _{A\hat{A}B\hat{B}}\right\Vert
_{2}\nonumber\\
&  =\left\Vert \int_{0}^{\infty}\text{d}\mu(t)\ tw_{AB\hat{A}\hat{B}}%
^{t}\left\vert \Gamma\right\rangle _{A\hat{A}B\hat{B}}\right\Vert _{2}\\
&  \leq\int_{0}^{T}\text{d}\mu(t)\ t\left\Vert w_{AB\hat{A}\hat{B}}%
^{t}\left\vert \Gamma\right\rangle _{A\hat{A}B\hat{B}}\right\Vert
_{2}\nonumber\\
&  \qquad+\left\Vert \int_{T}^{\infty}\text{d}\mu(t)\ t\ w_{AB\hat{A}\hat{B}%
}^{t}\left\vert \Gamma\right\rangle _{A\hat{A}B\hat{B}}\right\Vert _{2}%
\end{align}
Let us study the two terms separately. For the first term, from
Cauchy--Schwarz%
\begin{equation*}
\left\vert \int_{0}^{T}\text{d}\mu(t)f(t)g(t)\right\vert ^{2}
\leq\left[  \int_{0}^{T}\text{d}\mu(t)f^{2}(t)\right]  \left[  \int_{0}%
^{T}\text{d}\mu(t)g^{2}(t)\right]  ,
\end{equation*}
we have that%
\begin{align}
&  \left[  \int_{0}^{T}\text{d}\mu(t)\ t\left\Vert w_{AB\hat{A}\hat{B}}%
^{t}\left\vert \Gamma\right\rangle _{A\hat{A}B\hat{B}}\right\Vert _{2}\right]
^{2}\nonumber\\
&  \leq\mu([0,T])\int_{0}^{T}\text{d}\mu(t)\ t^{2}\left\Vert w_{AB\hat{A}%
\hat{B}}^{t}\left\vert \Gamma\right\rangle _{A\hat{A}B\hat{B}}\right\Vert
_{2}^{2}\\
&  \leq\mu([0,T])\int_{0}^{\infty}\text{d}\mu(t)\ t^{2}\left\Vert w_{AB\hat
{A}\hat{B}}^{t}\left\vert \Gamma\right\rangle _{A\hat{A}B\hat{B}}\right\Vert
_{2}^{2}\\
&  \leq\mu([0,T])\left[  Q_{f}(\rho_{AB}\Vert\sigma_{AB})-Q_{f}(\rho_{A}%
\Vert\sigma_{A})\right]  .
\end{align}
Moving to the second term, from the reasoning in the proof of \cite[Theorem~1.7]{CV17}, we find that for
any positive operator~$X$
\begin{align}
t\left(  \frac{1}{t}-\frac{1}{t+X}\right)   &  \leq t\left(  \frac{1}{t}%
-\frac{1}{t+\left\Vert X\right\Vert_{\infty} }\right)  I\\
&  =\frac{1}{1+t/\left\Vert X\right\Vert_{\infty} }I
\end{align}
so that
$
\int_{T}^{\infty}\text{d}\mu(t)\ t\left(  \frac{1}{t}-\frac{1}{t+X}\right)
\leq g(\left\Vert X\right\Vert ,T)\ I$.
This leads to the development in \eqref{eq:last-dev}, \begin{figure*}[ptb]%
\begin{align}
&  \left\Vert \int_{T}^{\infty}\text{d}\mu(t)\ t\ w_{AB\hat{A}\hat{B}}%
^{t}\left\vert \Gamma\right\rangle _{A\hat{A}B\hat{B}}\right\Vert
_{2}\nonumber\\
&  \leq\left\Vert \int_{T}^{\infty}\text{d}\mu(t)\ t\left[  \left(
t^{-1}-\left(  \Delta_{AB\hat{A}\hat{B}}+t\right)  ^{-1}\right)  \sigma
_{AB}^{1/2}-\sigma_{AB}^{1/2}\sigma_{A}^{-1/2}\left(  t^{-1}-\left(
\Delta_{A\hat{A}}+t\right)  ^{-1}\right)  \sigma_{A}^{1/2}\right]  \left\vert
\Gamma\right\rangle _{A\hat{A}B\hat{B}}\right\Vert _{2}\nonumber\\
&  \leq\left\Vert \int_{T}^{\infty}\text{d}\mu(t)\ t\left(  t^{-1}-\left(
\Delta_{AB\hat{A}\hat{B}}+t\right)  ^{-1}\right)  \sigma_{AB}^{1/2}\left\vert
\Gamma\right\rangle _{A\hat{A}B\hat{B}}\right\Vert _{2}+\left\Vert \int
_{T}^{\infty}\text{d}\mu(t)\ t\ \sigma_{AB}^{1/2}\sigma_{A}^{-1/2}\left(
t^{-1}-\left(  \Delta_{A\hat{A}}+t\right)  ^{-1}\right)  \sigma_{A}%
^{1/2}\left\vert \Gamma\right\rangle _{A\hat{A}B\hat{B}}\right\Vert
_{2}\nonumber\\
&  =\left\Vert \int_{T}^{\infty}\text{d}\mu(t)\ t\left(  t^{-1}-\left(
\Delta_{AB\hat{A}\hat{B}}+t\right)  ^{-1}\right)  \sigma_{AB}^{1/2}\left\vert
\Gamma\right\rangle _{A\hat{A}B\hat{B}}\right\Vert _{2}+\left\Vert \int
_{T}^{\infty}\text{d}\mu(t)\ t\left(  t^{-1}-\left(  \Delta_{A\hat{A}%
}+t\right)  ^{-1}\right)  \sigma_{A}^{1/2}\left\vert \Gamma\right\rangle
_{A\hat{A}}\right\Vert _{2}\nonumber\\
&  \leq g(\left\Vert \Delta_{AB\hat{A}\hat{B}}\right\Vert_{\infty} ,T)\left\Vert
\sigma_{AB}^{1/2}\left\vert \Gamma\right\rangle _{A\hat{A}B\hat{B}}\right\Vert
_{2}+g(\left\Vert \Delta_{A\hat{A}}\right\Vert_{\infty} ,T)\left\Vert \sigma_{A}%
^{1/2}\left\vert \Gamma\right\rangle _{A\hat{A}}\right\Vert _{2}\nonumber\\
&  =\left[  g(\left\Vert \Delta_{AB\hat{A}\hat{B}}\right\Vert_{\infty} ,T)+g(\left\Vert
\Delta_{A\hat{A}}\right\Vert_{\infty} ,T)\right]  \operatorname{Tr}\{\sigma_{A}%
\}\leq2g(\left\Vert \Delta_{AB\hat{A}\hat{B}}\right\Vert_{\infty} ,T)\operatorname{Tr}%
\{\sigma_{A}\}. \label{eq:last-dev}%
\end{align}
\hrulefill
\end{figure*} and after putting everything together, we get \eqref{eq:intermediate}.
\end{proof}

\section{Application to Holevo's just-as-good fidelity}

I now specialize the above analysis to the case of the operator convex function $-x^{\alpha}$ for
$\alpha\in(0,1)$, and I abbreviate the corresponding quasi-entropy as $Q_\alpha$. For this case, from \cite[Section~8]{HMPB11}, we have that
$
d\mu(t)=\frac{\sin(\alpha\pi)}{\pi}t^{\alpha-1}\ \text{d}t$.
Plugging into the quantities in Lemma~\ref{lem:quasi}, we find that%
\begin{align}
\mu([0,T])    =
\frac{\sin(\alpha\pi)}{\pi}\int_{0}^{T}t^{\alpha-1}\ \text{d}t
  =\frac{\sin(\alpha\pi)}{\alpha\pi}T^{\alpha}.
\end{align}
We also find that%
\begin{align}
&  g(\left\Vert \Delta_{AB\hat{A}\hat{B}}\right\Vert_{\infty} ,T)\nonumber\\
&  =\int_{T}^{\infty}\text{d}\mu(t)\ \frac{1}{1+t/\left\Vert \Delta_{AB\hat
{A}\hat{B}}\right\Vert_{\infty} }\\
&  =\frac{\sin(\alpha\pi)}{\pi}\int_{T}^{\infty}\text{d}t\ t^{\alpha-1}%
\frac{1}{1+t/\left\Vert \Delta_{AB\hat{A}\hat{B}}\right\Vert_{\infty} }\\
&  \leq\frac{\sin(\alpha\pi)}{\pi}\int_{T}^{\infty}\text{d}t\ t^{\alpha
-1}\frac{1}{t/\left\Vert \Delta_{AB\hat{A}\hat{B}}\right\Vert_{\infty} }\\
&  =\frac{\sin(\alpha\pi)\left\Vert \Delta_{AB\hat{A}\hat{B}}\right\Vert_{\infty} }{\pi
T^{1-\alpha}\left(  1-\alpha\right)  }.
\end{align}
Furthermore, we have that%
\begin{align}
\nu(X)  &  =\int_{0}^{\infty}\text{d}\mu(t)~t\left(  \frac{1}{t}-\frac{1}%
{t+X}\right) \\
&  =\frac{\sin(\alpha\pi)}{\pi}\int_{0}^{\infty}\text{d}t\ t^{\alpha
-1}t\left(  \frac{1}{t}-\frac{1}{t+X}\right) 
  =X^{\alpha}.
\end{align}
Substituting into \eqref{eq:intermediate}, we find that%
\begin{multline}
\left\Vert \left[  \sigma_{AB}^{1/2}\sigma_{A}^{-1/2}\Delta_{A\hat{A}}%
^{\alpha}\sigma_{A}^{1/2}-\Delta_{AB\hat{A}\hat{B}}^{\alpha}\sigma_{AB}%
^{1/2}\right]  \left\vert \Gamma\right\rangle _{A\hat{A}B\hat{B}}\right\Vert
_{2}\leq\label{eq:alpha-to-simplify}\\
\left[  \frac{\sin(\alpha\pi)}{\alpha\pi}T^{\alpha}\right]  ^{1/2}\left[
Q_{\alpha}(\rho_{AB}\Vert\sigma_{AB})-Q_{\alpha}(\rho_{A}\Vert\sigma
_{A})\right]  ^{1/2}\\
+\frac{2\sin(\alpha\pi)\left\Vert \Delta_{AB\hat{A}\hat{B}}\right\Vert_{\infty} }{\pi
T^{1-\alpha}\left(  1-\alpha\right)  }\operatorname{Tr}\{\sigma_{A}\}.
\end{multline}

We can consider this for an arbitrary $\alpha\in(0,1)$, but the most
interesting and physically relevant case seems to occur when $\alpha=1/2$. So
I now prove the claim in \eqref{eq:main-result}.

For $\alpha=1/2$, the lower bound in \eqref{eq:alpha-to-simplify}\ simplifies
to%
\begin{multline}
\left\Vert \left[  \sigma_{AB}^{1/2}\sigma_{A}^{-1/2}\rho_{A}^{1/2}-\rho
_{AB}^{1/2}\right]  \left\vert \Gamma\right\rangle _{A\hat{A}B\hat{B}%
}\right\Vert _{2}\\
=\left\Vert \sigma_{AB}^{1/2}\sigma_{A}^{-1/2}\rho_{A}^{1/2}-\rho_{AB}%
^{1/2}\right\Vert _{2},
\end{multline}
while the upper bound in \eqref{eq:alpha-to-simplify}\ becomes%
\begin{multline}
\left[  2/\pi\right]  ^{1/2}T^{1/4}\left[  Q_{1/2}(\rho_{AB}%
\Vert\sigma_{AB})-Q_{1/2}(\rho_{A}\Vert\sigma_{A})\right]  ^{1/2}\\
+\frac{4\left\Vert \Delta_{AB\hat{A}\hat{B}}\right\Vert_{\infty} }{\pi T^{1/2}%
}\operatorname{Tr}\{\sigma_{A}\}.
\end{multline}
Now minimizing over $T>0$ gives the choice%
\begin{equation}
T=\left[  \frac{8\left\Vert \Delta_{AB\hat{A}\hat{B}}\right\Vert_{\infty}
\operatorname{Tr}\{\sigma_{A}\}}{\left[  2\pi\right]  ^{1/2}\left[
Q_{1/2}(\rho_{AB}\Vert\sigma_{AB})-Q_{1/2}(\rho_{A}\Vert\sigma_{A})\right]
^{1/2}}\right]  ^{4/3},
\end{equation}
leading to the upper bound%
\begin{multline}
(3/[2^{2/3}])\Big[  (8/\pi^{2})  [Q_{1/2}(\rho
_{AB}\Vert\sigma_{AB})\\
-Q_{1/2}(\rho_{A}\Vert\sigma_{A})]\left\Vert \Delta_{AB\hat{A}\hat{B}%
}\right\Vert_{\infty} \operatorname{Tr}\{\sigma_{A}\}\Big]^{1/3}.
\end{multline}
Thus, the final inequality is%
\begin{multline}
\frac{\pi^{2}}{54\left\Vert \Delta_{AB\hat{A}\hat{B}}\right\Vert_{\infty}
\operatorname{Tr}\{\sigma_{A}\}}\left\Vert \sigma_{AB}^{1/2}\sigma_{A}%
^{-1/2}\rho_{A}^{1/2}-\rho_{AB}^{1/2}\right\Vert _{2}^{3}\\
\leq Q_{1/2}(\rho_{AB}\Vert\sigma_{AB})-Q_{1/2}(\rho_{A}\Vert\sigma_{A}).
\end{multline}
Using definitions, this is then equivalent to%
\begin{multline}
\frac{\pi^{2}}{54\left\Vert \Delta_{AB\hat{A}\hat{B}}\right\Vert_{\infty}
\operatorname{Tr}\{\sigma_{A}\}}\left\Vert \sigma_{AB}^{1/2}\sigma_{A}%
^{-1/2}\rho_{A}^{1/2}-\rho_{AB}^{1/2}\right\Vert _{2}^{3}\\
\leq\operatorname{Tr}\{\rho_{A}^{1/2}\sigma_{A}^{1/2}\}-\operatorname{Tr}%
\{\rho_{AB}^{1/2}\sigma_{AB}^{1/2}\}.
\end{multline}
The estimate from \cite[Lemma~2.2]{CV17} then gives
\begin{multline}
\frac{\pi^{2}/432}{\left\Vert \Delta_{AB\hat{A}\hat{B}}\right\Vert_{\infty}
\operatorname{Tr}\{\sigma_{A}\}}\left\Vert \sigma_{AB}^{1/2}\sigma_{A}%
^{-1/2}\rho_{A}\sigma_{A}^{-1/2}\sigma_{AB}^{1/2}-\rho_{AB}\right\Vert
_{1}^{3}\\
\leq\operatorname{Tr}\{\rho_{A}^{1/2}\sigma_{A}^{1/2}\}-\operatorname{Tr}%
\{\rho_{AB}^{1/2}\sigma_{AB}^{1/2}\}.
\end{multline}
Observe that $\left\Vert \Delta_{AB\hat{A}\hat{B}}\right\Vert_{\infty} =\left\Vert
\sigma_{AB}^{-1}\otimes\rho_{\hat{A}\hat{B}}^{T}\right\Vert_{\infty} \leq\frac
{1}{\lambda_{\min}(\sigma_{AB})}$ because $\rho_{AB}$ is a density operator.
So we then get
\begin{multline}
\frac{\pi^{2}}{54}\frac{\lambda_{\min}(\sigma_{AB})}{\operatorname{Tr}\{\sigma_{A}%
\}}\left\Vert \sigma_{AB}^{1/2}\sigma_{A}^{-1/2}\rho_{A}^{1/2}-\rho_{AB}%
^{1/2}\right\Vert _{2}^{3}\\
\leq\operatorname{Tr}\{\rho_{A}^{1/2}\sigma_{A}^{1/2}\}-\operatorname{Tr}%
\{\rho_{AB}^{1/2}\sigma_{AB}^{1/2}\},
\end{multline}%
\vspace{-.3in}
\begin{multline}
\frac{\pi^{2}}{432}\frac{\lambda_{\min}(\sigma_{AB})}{\operatorname{Tr}\{\sigma_{A}%
\}}\left\Vert \sigma_{AB}^{1/2}\sigma_{A}^{-1/2} \rho_{A} \sigma_{A}%
^{-1/2}\sigma_{AB}^{1/2}-\rho_{AB}\right\Vert _{1}^{3}\\
\leq\operatorname{Tr}\{\rho_{A}^{1/2}\sigma_{A}^{1/2}\}-\operatorname{Tr}%
\{\rho_{AB}^{1/2}\sigma_{AB}^{1/2}\},
\end{multline}
the latter of which being what was claimed in \eqref{eq:main-result}.



\textbf{Acknowledgements.} I thank Marco Piani and Anna Vershynina for discussions related to
the topic of this paper, and I acknowledge support from the NSF under grant no.~1714215.

\textbf{Note:} The results and proofs in  the main text of this paper were developed after \cite{CV17} but independently of arXiv:1710.08080 and were communicated privately by email in mid-October 2017.

\appendix

This appendix contains a generalization of the result in
\eqref{eq:main-result} to arbitrary quantum channels.

\begin{theorem}
Let $\omega$ be a density operator and $\tau$ a positive semi-definite
operator such that $\operatorname{supp}(\omega)\subseteq\operatorname{supp}%
(\tau)$. Let $\mathcal{N}$ be a quantum channel. Then%
\begin{multline}
\sqrt{F_{H}}(\mathcal{N}(\omega),\mathcal{N}(\tau))\geq \sqrt{F_{H}}(\omega,\tau)\\
+\frac{\pi^{2}}{432}\frac{\lambda_{\min}(\tau)}{\operatorname{Tr}\{\tau
\}}\left\Vert \mathcal{P}(\mathcal{N}(\omega))-\omega\right\Vert _{1}^{3},
\end{multline}
where $\lambda_{\min}(\tau)$ now denotes the minimum non-zero eigenvalue of
$\tau$ and $\mathcal{P}$ denotes the Petz recovery map for $\tau$ and
$\mathcal{N}$, defined~as%
\begin{equation}
\mathcal{P}(\cdot)=\tau^{1/2}\mathcal{N}^{\dag}[(\mathcal{N}(\tau
))^{-1/2}(\cdot)(\mathcal{N}(\tau))^{-1/2}]\tau^{1/2}.
\end{equation}

\end{theorem}

\begin{proof}
Let us start by returning to \eqref{eq:main-result} and reflecting on its
statement as well as its proof. If $\operatorname{supp}(\rho_{AB}%
)\subseteq\operatorname{supp}(\sigma_{AB})$, then without loss of generality,
we can restrict the whole space of systems $A$ and $B$ to the support of
$\sigma_{AB}$ and the inequality in \eqref{eq:main-result} holds with
$\lambda_{\min}(\sigma_{AB})$ equal to the minimum non-zero eigenvalue of
$\sigma_{AB}$. Now we can apply this result, as well as the well known
Stinespring dilation theorem, in order to arrive at the statement of the
theorem. Stinespring's theorem states that for a quantum channel $\mathcal{N}$
acting on a state $\omega$ of a system~$S$, there exists an isometry
$U_{S\rightarrow AB}$ such that%
\begin{equation}
\mathcal{N}(\omega)=\operatorname{Tr}_{B}\{U_{S\rightarrow AB}\omega
(U_{S\rightarrow AB})^{\dag}\}.
\end{equation}
So we pick%
\begin{align}
\rho_{AB}  & =U_{S\rightarrow AB}\omega(U_{S\rightarrow AB})^{\dag},\\
\sigma_{AB}  & =U_{S\rightarrow AB}\tau(U_{S\rightarrow AB})^{\dag},
\end{align}
so that
$
\rho_{A}   =\mathcal{N}(\omega)$, 
$\sigma_{A}   =\mathcal{N}(\tau)$,
and then find that%
\begin{multline}
\sqrt{F_{H}}(\mathcal{N}(\omega),\mathcal{N}(\tau))\geq \\
\sqrt{F_{H}}(U_{S\rightarrow
AB}\omega(U_{S\rightarrow AB})^{\dag},U_{S\rightarrow AB}\tau
(U_{S\rightarrow AB})^{\dag})\\
+\frac{\pi^{2}}{432}\frac{\lambda_{\min}(U_{S\rightarrow AB}\tau
(U_{S\rightarrow AB})^{\dag})}{\operatorname{Tr}\{\mathcal{N}(\tau)\}}\times\\
\left\Vert \mathcal{R}^{\sigma}_{A\rightarrow AB}(\mathcal{N}(\omega))-U_{S\rightarrow
AB}\omega(U_{S\rightarrow AB})^{\dag}\right\Vert _{1}^{3}.
\end{multline}
Due to isometric invariance of Holevo's just-as-good fidelity and the minimum
non-zero eigenvalue, and the fact that $\mathcal{N}$ is trace preserving, we
find that%
\begin{align}
F_{H}(U_{S\rightarrow AB}\omega(U_{S\rightarrow AB})^{\dag},U_{S\rightarrow
AB}\omega(U_{S\rightarrow AB})^{\dag})  & =F_{H}(\omega,\tau), \nonumber \\
\lambda_{\min}(U_{S\rightarrow AB}\tau(U_{S\rightarrow AB})^{\dag})  &
=\lambda_{\min}(\tau), \nonumber \\
\operatorname{Tr}\{\mathcal{N}(\tau)\}  & =\operatorname{Tr}\{\tau\}.
\end{align}
Also, the Petz map
$\mathcal{R}^{\sigma}_{A\rightarrow AB}$
simplifies for our choices as%
\begin{align}
& \mathcal{R}^{\sigma}_{A\rightarrow AB}(\cdot)\nonumber\\
& =\sigma_{AB}^{1/2}[\sigma_{A}^{-1/2}(\cdot)\sigma_{A}^{-1/2}\otimes
I_{B}]\sigma_{AB}^{1/2}\nonumber \\
& =[U_{S\rightarrow AB}\tau(U_{S\rightarrow AB})^{\dag}]^{1/2}[\mathcal{N}%
(\tau)^{-1/2}(\cdot)\mathcal{N}(\tau)^{-1/2}\otimes I_{B}]\nonumber\\
& \qquad \times\lbrack U_{S\rightarrow AB}\tau(U_{S\rightarrow AB})^{\dag}]^{1/2}\nonumber \\
& =U_{S\rightarrow AB}\tau^{1/2}(U_{S\rightarrow AB})^{\dag}[\mathcal{N}%
(\tau)^{-1/2}(\cdot)\mathcal{N}(\tau)^{-1/2}\otimes I_{B}]\nonumber\\
& \qquad \times U_{S\rightarrow AB}\tau^{1/2}(U_{S\rightarrow AB})^{\dag}\nonumber \\
& =U_{S\rightarrow AB}\tau^{1/2}\mathcal{N}^{\dag}[\mathcal{N}(\tau
)^{-1/2}(\cdot)\mathcal{N}(\tau)^{-1/2}]\tau^{1/2}(U_{S\rightarrow AB})^{\dag
}.
\end{align}
Isometric invariance of the trace norm and the above then gives
\begin{multline}
\left\Vert \mathcal{R}^{\sigma}_{A\rightarrow AB}(\mathcal{N}(\omega))-U_{S\rightarrow
AB}\omega(U_{S\rightarrow AB})^{\dag}\right\Vert _{1}\\
=\left\Vert \mathcal{P}(\mathcal{N}(\omega))-\omega\right\Vert _{1},
\end{multline}
concluding the proof.
\end{proof}

\bibliographystyle{IEEEtran}
\bibliography{Ref}

\end{document}